\documentclass[aps,prl,reprint,superscriptaddress,nofootinbib]{revtex4-2}

\usepackage[T1]{fontenc}
\usepackage[utf8]{inputenc}
\usepackage{amsmath,amssymb,amsthm}
\usepackage{bm}
\usepackage[colorlinks=true,citecolor=blue,linkcolor=blue,urlcolor=blue]{hyperref}

\allowdisplaybreaks

\newtheorem{theorem}{Theorem}

\newtheorem{definition}{Definition}

\begin{document}

\title{Ontic Dynamical Locality Reduces to Bell Locality}

\author{Ming Yang}
\date{\today}

\begin{abstract}
Bell inequalities exclude a broad class of local hidden-variable explanations of
quantum correlations. A recurring objection is that the usual Bell form is
static, whereas real measuring devices may contain local memory, stochastic
dynamics, and measurement-induced disturbances of their hidden variables. We
formulate this objection as a general transition-kernel model for dynamical
hidden variables. The only locality assumption is imposed at the ontic level:
conditional on the pre-measurement ontic state, the transition kernel and
response function in each wing depend on the local setting but not on the
distant setting or on distant post-measurement variables. Under measurement
independence, all such dynamics can be absorbed into effective local response
functions. The resulting probabilities have exactly the static Bell-local form
and therefore obey the CHSH inequality. The result identifies the available
escape routes for reproducing quantum correlations: violating ontic dynamical
locality, violating measurement independence, or abandoning a classical
hidden-variable ontology. As a consequence, local classical dynamical complexity
cannot by itself spoof Bell-nonlocal statistics in device-independent protocols.
\end{abstract}

\maketitle

\emph{Introduction.---}
Bell's theorem shows that quantum correlations cannot be reproduced by the
standard local hidden-variable model \cite{Bell1964,Bell1971,CHSH1969}. This
conclusion is now not only a foundational statement, but also the basis of
device-independent quantum information, where Bell violation certifies
properties of devices without trusting their microscopic implementation
\cite{Brunner2014,Pironio2010}.

A persistent concern is that the Bell model appears static. Real apparatuses
are dynamical systems: they may contain internal memory, hidden fields,
chaotic variables, and local measurement-induced disturbances. This motivates a
natural question. Can a completely local dynamical process, by reshuffling
hidden variables during the measurement, generate correlations that would be
impossible in a static Bell model? Similar intuitions underlie several proposed
classical or macroscopic analogues of Bell correlations
\cite{Matzkin2008,HessPhilipp2001,HessPhilipp2004,Vervoort2018,Papatryfonos2022}.

The answer is no, provided locality is stated at the correct level. Operational
no-signaling alone is too weak: many nonlocal hidden-variable theories are
no-signaling after averaging. This distinction is familiar from Bell's local
causality, Shimony's discussion of controllable and uncontrollable nonlocality,
and the separation between parameter independence and outcome independence
\cite{Bell1976,Shimony1986,Jarrett1984,Norsen2009}. We therefore avoid the
ambiguous slogan ``local realism'' and formulate the assumptions directly in
terms of transition kernels over ontic states. The theorem below shows that any
measurement-independent dynamical model satisfying ontic dynamical locality is
mathematically reducible to an ordinary static Bell-local model.

\emph{Dynamical hidden-variable model.---}
Let $x$ and $y$ denote the measurement settings chosen in two spacelike
separated wings, and let $a$ and $b$ denote the corresponding outcomes. Before
the measurement interaction, the experiment is described by a complete
pre-measurement ontic state $\Lambda$ distributed according to
$\rho_{\rm pre}(d\Lambda)$. Measurement independence is the condition
\begin{equation}
  \rho_{\rm pre}(d\Lambda|x,y)=\rho_{\rm pre}(d\Lambda).
  \label{eq:mi}
\end{equation}
The state $\Lambda$ may include local variables, shared source variables,
apparatus variables, and any common classical field value available before the
settings are applied.

During measurement, Alice's apparatus may generate an arbitrary local
post-measurement variable $u_A$ according to a stochastic kernel
$K_A(du_A|\Lambda,x)$, and Bob's apparatus may generate $u_B$ according to
$K_B(du_B|\Lambda,y)$. The variables $u_A$ and $u_B$ are deliberately general:
they may contain disturbed local hidden variables, memory registers, local
copies of a shared variable, or locally updated representatives of a background
field. Outcomes are then produced by response functions
$P_A(a|x,\Lambda,u_A)$ and $P_B(b|y,\Lambda,u_B)$.

\begin{definition}[Ontic dynamical locality]
A dynamical hidden-variable model is ontically dynamically local if, conditional
on the pre-measurement ontic state $\Lambda$, its measurement-stage probability
factorizes as
\begin{align}
 &\Pr(a,b,du_A,du_B|x,y,\Lambda) \nonumber\\
 &\quad =
 K_A(du_A|\Lambda,x)K_B(du_B|\Lambda,y) \nonumber\\
 &\qquad \times
 P_A(a|x,\Lambda,u_A)P_B(b|y,\Lambda,u_B).
 \label{eq:odl}
\end{align}
Here the left-hand side is a conditional measure in the post-measurement
variables $u_A,u_B$ for fixed discrete outcomes $a,b$.
Equivalently, neither wing's transition kernel or response function depends on
the distant setting or on the distant post-measurement variable, once
$\Lambda$ is specified.
\end{definition}

This condition is stronger than operational no-signaling. It is the dynamical
version of Bell's local causality: all correlations between the wings must be
screened off by the complete pre-measurement ontic state. It may be motivated
by Einstein's separability intuition, but the mathematical assumption is
Eq.~\eqref{eq:odl}, not the historical phrase itself.

\emph{Equivalence theorem.---}
Combining the preparation and measurement stages gives
\begin{align}
 P(a,b|x,y)
  ={}&\int \rho_{\rm pre}(d\Lambda)
      \int K_A(du_A|\Lambda,x)P_A(a|x,\Lambda,u_A) \nonumber\\
    &\times
      \int K_B(du_B|\Lambda,y)P_B(b|y,\Lambda,u_B).
 \label{eq:joint}
\end{align}

\begin{theorem}[Dynamical-to-static reduction]
Every measurement-independent dynamical hidden-variable model satisfying
ontic dynamical locality is equivalent to a static Bell-local model.
Consequently, for binary outcomes and two settings per party, it satisfies the
CHSH inequality $|S|\leq 2$.
\end{theorem}

\begin{proof}
Define effective response functions
\begin{align}
 Q_A(a|x,\Lambda)
   &:= \int K_A(du_A|\Lambda,x)P_A(a|x,\Lambda,u_A), \label{eq:qa}\\
 Q_B(b|y,\Lambda)
   &:= \int K_B(du_B|\Lambda,y)P_B(b|y,\Lambda,u_B). \label{eq:qb}
\end{align}
Because $K_A$ and $K_B$ are normalized probability kernels and $P_A,P_B$ are
normalized response functions, $Q_A$ and $Q_B$ are valid conditional
probabilities. Substituting Eqs.~\eqref{eq:qa} and \eqref{eq:qb} into
Eq.~\eqref{eq:joint} yields
\begin{equation}
  P(a,b|x,y)
  =\int \rho_{\rm pre}(d\Lambda)
       Q_A(a|x,\Lambda)Q_B(b|y,\Lambda).
  \label{eq:bellform}
\end{equation}
Equation \eqref{eq:bellform} is precisely the static Bell-local form. The CHSH
bound follows by the usual convexity argument for local probabilistic response
functions.
\end{proof}

The proof contains the central physical point. Local measurement dynamics do
not disappear; they are absorbed into the effective response functions
$Q_A$ and $Q_B$. Since $\rho_{\rm pre}$ is independent of $(x,y)$, the
absorbed model remains Bell local. Local stochasticity, local chaos, local
memory, and local measurement disturbance therefore cannot generate Bell
nonlocality.

\emph{What about a disturbed shared variable?---}
A frequent objection is that the measurement could disturb a variable that is
``global'' or shared. The distinction is simple. If locality is maintained,
then a measurement in Alice's wing can only create a local post-measurement
representative of that variable, included in $u_A$; Bob's wing analogously
creates a representative included in $u_B$. These local representatives are
already covered by Eq.~\eqref{eq:odl} and are absorbed into $Q_A$ and $Q_B$.

By contrast, suppose there is a single indivisible post-measurement variable
$G'$ that is updated by Alice's setting and is also available to Bob's response
function. Then Bob's response has the form
$P_B(b|y,x,\Lambda)$, or equivalently depends on a variable carrying
Alice's setting. This is not a local dynamical refinement of Bell's model; it
is an ontic nonlocal dependence. If Bob is completely insensitive to the
Alice-dependent part of $G'$, that part is operationally redundant for Bob and
can be absorbed into Alice's local variable. Thus a disturbed shared variable
either reduces to the local-copy case or violates ontic dynamical locality.

\emph{Trichotomy.---}
The reduction identifies the precise ways to evade the theorem. A model that
reproduces quantum Bell correlations must give up at least one of the following:

\begin{enumerate}
\item \emph{Ontic dynamical locality.} The distant setting, distant
post-measurement variable, or distant outcome-relevant structure enters a
local response at the hidden-variable level. Bohmian mechanics belongs here:
although it is operationally no-signaling in quantum equilibrium, the guiding
wavefunction produces nonlocal dependence in the underlying dynamics
\cite{Valentini1991,Norsen2009}.
\item \emph{Measurement independence.} The pre-measurement distribution is
contextual, $\rho_{\rm pre}(d\Lambda|x,y)\neq\rho_{\rm pre}(d\Lambda)$. This
includes superdeterministic models and any protocol in which a background
field is allowed to equilibrate with the future measurement settings before
statistics are sampled.
\item \emph{A classical hidden-variable ontology.} Quantum probability theory
uses noncommutative operator algebras rather than a single underlying
Kolmogorov probability space carrying pre-existing values for all relevant
measurement contexts. Relational or other holistic ontologies can also reject
the classical hidden-variable premise rather than violating the theorem's
dynamical assumptions \cite{Rovelli1996}.
\end{enumerate}

This trichotomy is narrower and cleaner than the phrase ``local realism.'' The
mathematics does not require a philosophical realism premise by name; it
requires an ontic state space, measurement independence, and the local
factorization in Eq.~\eqref{eq:odl}.

\emph{Diagnostic applications.---}
The theorem provides a diagnostic for proposed classical dynamical analogues.
If a model has only local transition kernels conditioned on a
setting-independent pre-measurement state, it is already Bell local by
Eq.~\eqref{eq:bellform}. Historical temporal-disturbance and clock-variable
constructions therefore cannot evade Bell inequalities without slipping in a
failure of one of the assumptions \cite{Gill2002,Nijhoff2008}. Similarly,
hydrodynamic pilot-wave analogues can reproduce selected quantum-like
phenomena, but Bell-nonlocal statistics require either nonlocal effective
dynamics or a setting-dependent preparation of the relevant field state
\cite{Papatryfonos2022}.

The same point applies to device-independent protocols. A black box may contain
arbitrary classical memory and complicated local algorithms. Nevertheless, if
its internal dynamics satisfy Eq.~\eqref{eq:odl} and its pre-measurement state
is independent of the settings, it cannot produce a Bell violation. This is
not a full security proof for a complete experimental implementation; it rules
out a specific class of classical local dynamical spoofing mechanisms.

\emph{Conclusion.---}
Local dynamical complexity does not enlarge the set of Bell-local
correlations. Once the complete pre-measurement ontic state is conditioned on,
any local transition dynamics can be absorbed into static effective response
functions. A successful hidden-variable account of quantum Bell correlations
must therefore introduce ontic nonlocal dependence, measurement dependence, or
depart from a classical hidden-variable ontology. Dynamics alone is not an
escape from Bell locality.

\clearpage
\onecolumngrid
\setcounter{equation}{0}
\renewcommand{\theequation}{S\arabic{equation}}
\begin{center}
{\large\bf Supplemental Material}\\[0.5em]
{\bf Ontic Dynamical Locality Reduces to Bell Locality}
\end{center}

\appendix

\section{Measure-theoretic form of the reduction}

Let $(\Omega,\mathcal{F})$ be the measurable space of pre-measurement ontic
states, and let $(U_A,\mathcal{U}_A)$ and $(U_B,\mathcal{U}_B)$ be the
measurable spaces of local post-measurement variables. Assume these are
standard Borel spaces, so that regular conditional probabilities and stochastic
kernels are well behaved \cite{Bogachev2007,Kechris1995}. For each setting
$x$, $K_A(\cdot|\Lambda,x)$ is a probability kernel from $\Omega$ to $U_A$;
similarly $K_B$ is a probability kernel from $\Omega$ to $U_B$.

For bounded response functions $P_A$ and $P_B$, Tonelli's theorem justifies the
exchange and grouping of integrals:
\begin{align}
 P(a,b|x,y)
  ={}&\int_{\Omega}\rho_{\rm pre}(d\Lambda)
      \int_{U_A}K_A(du_A|\Lambda,x)
      \int_{U_B}K_B(du_B|\Lambda,y) \nonumber\\
    &\times P_A(a|x,\Lambda,u_A)P_B(b|y,\Lambda,u_B) \nonumber\\
  ={}&\int_{\Omega}\rho_{\rm pre}(d\Lambda)
      Q_A(a|x,\Lambda)Q_B(b|y,\Lambda).
\end{align}
Normalization follows immediately:
\begin{align}
 \sum_a Q_A(a|x,\Lambda)
   &=\int K_A(du_A|\Lambda,x)
       \sum_a P_A(a|x,\Lambda,u_A) \nonumber\\
   &=1,
\end{align}
and likewise for $Q_B$.

\section{Dynamic shared variables}

Suppose $\Lambda=(\Lambda_A,\Lambda_B,\Lambda_g)$, where $\Lambda_g$ is a
pre-measurement shared component. A local measurement may produce
$u_A=(\lambda'_A,\gamma_A)$ and $u_B=(\lambda'_B,\gamma_B)$, where
$\gamma_A$ and $\gamma_B$ are locally updated representatives of the shared
component. The most general local measurement stage is then
\begin{align}
 P(a,b|x,y)
  ={}&\int \rho_{\rm pre}(d\Lambda)
      \int K_A(d\lambda'_A,d\gamma_A|\Lambda,x)
      P_A(a|x,\Lambda,\lambda'_A,\gamma_A) \nonumber\\
    &\times
      \int K_B(d\lambda'_B,d\gamma_B|\Lambda,y)
      P_B(b|y,\Lambda,\lambda'_B,\gamma_B).
\end{align}
This is only a relabeling of Eq.~\eqref{eq:joint}; it reduces to the Bell form
by the same construction.

If instead a single post-measurement variable $\Gamma'$ is jointly available to
both wings and changes with Alice's setting, then Bob's response is not of the
form $P_B(b|y,\Lambda,u_B)$. Either Bob's outcome probabilities depend on the
Alice-dependent part of $\Gamma'$, which violates ontic dynamical locality, or
they do not, in which case that part belongs to Alice's local effective
variable for the purpose of all observable probabilities.

\section{General perturbations of a shared component}

This section spells out the reduction when a measurement is said to perturb a
shared component. The point is not that shared components cannot be physically
disturbed. The point is that, under ontic dynamical locality, the
outcome-relevant part of any such disturbance either remains local to the wing
where it is generated or else introduces a nonlocal dependence.

Let the pre-measurement ontic state be
$\Lambda=(\Lambda_A,\Lambda_B,\Lambda_g)$, and suppose Alice's apparatus
generates a variable
\begin{equation}
  u_A=(\lambda'_A,\gamma_A)
\end{equation}
from a kernel $K_A(d\lambda'_A,d\gamma_A|\Lambda,x)$, while Bob generates
\begin{equation}
  u_B=(\lambda'_B,\gamma_B)
\end{equation}
from $K_B(d\lambda'_B,d\gamma_B|\Lambda,y)$. Here $\gamma_A$ and $\gamma_B$
may be interpreted as locally perturbed representatives of the same initial
shared component $\Lambda_g$. Ontic dynamical locality requires the joint
measurement-stage kernel to factorize:
\begin{align}
 &\Pr(a,b,d\lambda'_A,d\gamma_A,d\lambda'_B,d\gamma_B|\Lambda,x,y) \nonumber\\
 &\quad =
 K_A(d\lambda'_A,d\gamma_A|\Lambda,x)
 K_B(d\lambda'_B,d\gamma_B|\Lambda,y) \nonumber\\
 &\qquad \times
 P_A(a|x,\Lambda,\lambda'_A,\gamma_A)
 P_B(b|y,\Lambda,\lambda'_B,\gamma_B).
 \label{eq:global-copy-factor}
\end{align}
Equation \eqref{eq:global-copy-factor} is already in the form covered by the
main theorem, with enlarged local variables
$u_A=(\lambda'_A,\gamma_A)$ and $u_B=(\lambda'_B,\gamma_B)$. Therefore
\begin{align}
 Q_A(a|x,\Lambda)
 &=\int K_A(d\lambda'_A,d\gamma_A|\Lambda,x)
   P_A(a|x,\Lambda,\lambda'_A,\gamma_A), \\
 Q_B(b|y,\Lambda)
 &=\int K_B(d\lambda'_B,d\gamma_B|\Lambda,y)
   P_B(b|y,\Lambda,\lambda'_B,\gamma_B),
\end{align}
and the observable distribution is again
\begin{equation}
 P(a,b|x,y)=\int \rho_{\rm pre}(d\Lambda)
 Q_A(a|x,\Lambda)Q_B(b|y,\Lambda).
\end{equation}

Now consider the stronger claim that Alice's setting produces a single
post-measurement shared object $\Gamma'_x$ that is also available to Bob. Then
Bob's response must be written as
\begin{equation}
  P_B(b|y,\Lambda,\Gamma'_x).
\end{equation}
There are only two cases.

First, if for some $b,y,\Lambda$ the value of $P_B(b|y,\Lambda,\Gamma'_x)$
changes with $x$, Bob's response function depends on Alice's distant setting
at the ontic level. This violates Eq.~\eqref{eq:odl}. The model may still be
operationally no-signaling after averaging, but it is not dynamically Bell
local.

Second, if
\begin{equation}
  P_B(b|y,\Lambda,\Gamma'_x)=P_B(b|y,\Lambda)
\end{equation}
for all $b,y,\Lambda$ and all locally generated values of $\Gamma'_x$, then the
Alice-dependent part of $\Gamma'_x$ is irrelevant to Bob's outcome statistics.
It can be absorbed into Alice's local variable without changing any observable
probability. In this case the model reduces to the local-representative form
above and hence to the static Bell form.

Thus a dynamically perturbed shared component is not a fourth escape route. It
is either a pair of local representatives, which is reducible, or a single
distantly accessible object carrying the remote setting, which violates ontic
dynamical locality.

\section{Relation to Bohmian mechanics}

Bohmian mechanics is a useful boundary case because it reproduces quantum
correlations while preserving operational no-signaling in quantum equilibrium.
It therefore illustrates why the theorem is formulated in terms of ontic
dynamical locality rather than merely observed no-signaling.

For two particles, the Bohmian ontic state is usually written as
\begin{equation}
  \lambda=(\Psi,q_A,q_B),
\end{equation}
where $\Psi(q_A,q_B,t)$ is the wavefunction on configuration space and
$(q_A,q_B)$ are the actual particle positions. The velocity field for Bob's
particle is determined by the full configuration-space wavefunction, for
example
\begin{equation}
  \dot q_B=\frac{\hbar}{m_B}
  \operatorname{Im}
  \frac{\nabla_B\Psi(q_A,q_B,t)}{\Psi(q_A,q_B,t)}.
\end{equation}
For an entangled state, a change in Alice's measurement context can change the
effective global wavefunction guiding the joint configuration. Conditional on
the full ontic state, Bob's later trajectory and hence his outcome-relevant
dynamics can depend on Alice's distant setting. In the notation of the main
text, the response is not generally of the form
$P_B(b|y,\Lambda,u_B)$; it has the ontic dependence
\begin{equation}
  P_B(b|y,x,\Lambda)\neq P_B(b|y,\Lambda).
\end{equation}
This is a violation of ontic dynamical locality, not a counterexample to the
reduction theorem.

The reason this does not normally enable controllable superluminal signaling is
the quantum equilibrium condition
\begin{equation}
  \rho(q_A,q_B)=|\Psi(q_A,q_B)|^2.
\end{equation}
After averaging over the equilibrium distribution, Bob's marginal probabilities
agree with the quantum no-signaling theorem. In density-matrix language,
Alice's local choice does not change Bob's reduced state in a way that permits
signaling, so
\begin{equation}
  P_{\rm macro}(b|x,y)=P_{\rm macro}(b|y).
\end{equation}
Valentini's nonequilibrium analysis emphasizes precisely this distinction:
outside quantum equilibrium, the hidden nonlocal dependence need not be
statistically masked \cite{Valentini1991}. Therefore Bohmian mechanics falls
under the first branch of the trichotomy: it retains a hidden-variable ontology
and measurement independence, but it violates ontic dynamical locality.

\end{document}